\documentclass[12pt]{article}

\usepackage{amsmath,amsfonts,amsbsy,amsthm}

\usepackage{fullpage}
\usepackage{dsfont}
\usepackage{verbatim} 
\usepackage{mathtools}
\usepackage{xy}
\usepackage{hyperref}
\usepackage{changepage}
\usepackage{enumerate}
\usepackage{enumitem}
\setlist{leftmargin=*}

\numberwithin{equation}{section}

\makeatletter
\newtheoremstyle{corsivo}
{\medskipamount}{\medskipamount}%
{\itshape}{}%
{\bfseries}{}%
{ }
{\thmname{#1}\thmnumber{\@ifnotempty{#1}{ }\@upn{#2}}%
	\thmnote{ {\bfseries(#3)}}.}%
\makeatother

\theoremstyle{corsivo}
\newtheorem{thm}{Theorem}[section]
\newtheorem{lemma}[thm]{Lemma}

\newtheorem{prop}[thm]{Proposition}

\makeatletter
\newtheoremstyle{dritto}
{\medskipamount}{\medskipamount}%
{\rmfamily}{}%
{\bfseries}{}%
{ }
{\thmname{#1}\thmnumber{\@ifnotempty{#1}{ }\@upn{#2}}%
	\thmnote{ {\bfseries(#3)}}.}%
\makeatother

\theoremstyle{dritto}

\newtheorem{example}[thm]{Example}

\newtheorem{dfn}[thm]{Definition}
\newtheorem{rmk}[thm]{Remark}

\newcommand{\ga}{\gamma}

\newcommand{\eps}{\varepsilon}

\newcommand{\ph}{\varphi}

\newcommand{\Id}{\mathds{1}}  

\newcommand{\di}{\mathrm{d}}

\newcommand{\N}{\mathbb{N}}
\newcommand{\Z}{\mathbb{Z}}

\newcommand{\R}{\mathbb{R}}
\newcommand{\C}{\mathbb{C}}

\newcommand{\BR}{\mathcal{B}(L^2(\R^d))}
\newcommand{\D}{\mathfrak{D}}

\newcommand{\A}{\mathcal{A}}

\newcommand{\Hi}{\mathcal{H}}

\newcommand{\norm}[1]{\left\| #1 \right\|}

\newcommand{\bra}[1]{\left\langle #1 \right|}
\newcommand{\ket}[1]{\left| #1 \right\rangle}

\newcommand{\set}[1]{ \left\{  #1 \right\}} 

\DeclareMathOperator{\Ran}{Ran}

\DeclareMathOperator{\supp}{supp}

\DeclareMathOperator{\dist}{dist}

\newcommand{\ie}{{\sl i.\,e.\ }}   
\newcommand{\eg}{{\sl e.\,g.\ }} 

\newcommand{\virg}[1]{``#1''}

\renewcommand{\(}{\left(}
\renewcommand{\)}{\right)}

\renewcommand{\endrmk}{\hfill $\diamond$}

\begin{document}
\title{Algebraic localization of generalized Wannier \\ bases implies
	 Roe triviality in any dimension \vspace{10pt}}
\date{}
\author{Vincenzo Rossi and Gianluca Panati}
\maketitle

\begin{abstract}
    With the aim of understanding the localization-topology correspondence for \emph{non-periodic} gapped quantum systems,  we investigate the relation between the existence of an algebraically well-localized generalized Wannier basis and the topological triviality of the corresponding projection operator.  Inspired by the work of M.~Ludewig and G.C.~Thiang, we consider the triviality of a projection in the sense of \virg{coarse geometry}, i.e. as triviality in the $K_0$-theory of the Roe $C^*$-algebra of $\R^d$.  We obtain in Theorem \ref{original} a threshold, depending on the dimension, for the decay rate of the generalized Wannier functions which implies topological triviality in Roe's sense. This threshold reduces, for $d=2$, to the almost-optimal threshold appearing in the Localization Dichotomy Conjecture.
\end{abstract}
\section{Introduction}

Wannier bases, and their generalization to non-periodic systems, have been for decades a popular tool in solid state physics, both for practical and conceptual reasons.  On the first side, the existence of a well-localized (composite) Wannier basis (WB) allows computational algorithms whose cost scales only linearly with the system size \cite{Goedecker}. 
On the conceptual side, well-localized WBs allow a theoretical understanding of several phenomena in crystalline solids as e.g. macroscopic polarization and orbital magnetization in insulators \cite{Ceresoli}.   
We refer to \cite{Marzari} for a comprehensive review about WBs for solid state systems. 

More recently, with the discovery of topological insulators, WBs gained a new role as a tool to distinguish between ordinary and Chern insulators. A correspondence between the topological non-triviality of the vector bundle of occupied states (Bloch bundle) and the non-existence of a composite Wannier basis (CWB) with finite second moment of the position operator, has been noticed and proved in \cite{Monaco}, for periodic systems in dimension $d=2$ and $d=3$. More precisely, for any gapped periodic Schrödinger operator, the following  \emph{Localization Dichotomy} holds true: 
\begin{enumerate}[label=(\roman*)]
\item either there exists a CWB which is exponentially localized, and, correspondingly, the Chern class of the Bloch bundle is zero; or 
\item any possible choice of a CWB yields an infinite expectation value of the squared position operator. Explicitly, this means 
$$
\sum_{a=1}^m \int_{\R^d} \langle x - \gamma\rangle^2 |w_{\gamma, a}(x)|^2 \, dx = + \infty    \quad \forall \gamma \in \Gamma,  
 $$  
where $w_{\gamma, a}(x)=w_a(x -\gamma)$ are the elements of a CWB, $\Gamma$ is the periodicity lattice, and $a\in\set{1,\ldots, m}$ is an index corresponding to the number $m$ of relevant Bloch bands.
\end{enumerate}

After the appearance of \cite{Monaco}, it became clear that such a profound relation between localization and topology cannot depend on periodicity in an essential way, despite the fact that both terms of the correspondence (the Bloch bundle and, respectively, the Wannier functions) are defined by exploiting periodicity. 

The generalization of the Localization Dichotomy to non-periodic systems is grounded on the definition of Generalized Wannier Basis (GWB), which traces back to the pioneering work of Kohn and Onfroy \cite{Kohn} and Kivelson \cite{Kivelson}, and has been formulated in clear mathematical terms by A.~Nenciu and G.~Nenciu \cite{Nenciu3} (see Definition \ref{GWB} in the next Section). 
In \cite{Marcelli}, the existence of a localized GWB with finite second moment for the position operator 
has been conjecturally related to the vanishing of the Chern marker \cite{Avron2, Bellissard}. Moreover, the authors proved that the algebraic decay of a GWB implies the vanishing of the Chern marker, however for a degree of localization $s$ far from the optimal threshold $s\geq 1$.  (With this notation, as in Definition \ref{GWB}, $s=1$ corresponds to a finite second moment for the position operator). 
An almost-optimal threshold has been later reached by Lu and Stubbs \cite{Lu, Lu2}, who proved a similar result for every $s >1$ in dimension $d=2$.

A more general concept of Wannier basis has been proposed by Prodan \cite{Prodan}, and has been later associated to the name Ultra Generalized Wannier Basis (UGWB). While the latter bases might be useful from the computational viewpoint as they can be constructed for any magnetic Schr\"{o}dinger operator, the price to pay for such a generality is that UGWBs are not able to distinguish between ordinary and Chern insulators, as recently shown in \cite{Moscolari3}.     

Following an independent line of thought, several researchers approached the same problem from the viewpoint of $C^*$-algebras and Hilbert modules \cite{Ludewig, Bourne, Large scale}.  
In particular, the latter paper suggests that the topological triviality of a projection may be better investigated 
from the viewpoint of \virg{coarse geometry}, \ie by considering the corresponding class in the $K$-theory of the Roe $C^*$-algebra of the position space. Indeed, the authors of \cite{Large scale} proved that the existence of a 
Wannier basis which is $s$-localized for all $s>0$ in the sense of Definition \ref{GWB}, with the $L^{2}$-decay condition replaced by the $L^{\infty}$-one,  implies the triviality of the corresponding projection in the $K_0$-group of the Roe $C^*$-algebra, in a very general setting. 

In this paper, we combine ideas and methods from both lines of thoughts. Inspired by \cite{Large scale}, we prove that, when the position space is $\R^d$ for $d\geq 1$, the existence of a GWB with $L^2$-localization parameter $s>\frac{d}{2}$ implies the triviality of the corresponding projection in the $K_0$-group of the Roe $C^*$-algebra (Theorem \ref{original}). 
While our setting is much restricted in comparison with \cite{Large scale}, where general Riemannian manifolds are considered, our hypotheses are weaker, as only $s$-localization for $s>\frac{d}{2}$ is assumed. 
In comparison with Lu and Stubbs \cite{Lu, Lu2}, we obtain the same almost-optimal threshold $s >1$ in $\R^2$, but we have a result which is also valid in any dimension $d\geq 1$. Moreover, the structure of the proof emphasizes the role of the Roe $C^*$-algebra to describe the \virg{large scale} geometric properties which are relevant to physics, as suggested by Ludewig and Thiang.   Clearly, it remains to clarify under which conditions the triviality in the sense of Roe proved in this paper is equivalent to the Chern triviality considered in \cite{Lu2, Marcelli, Moscolari2}. While we are confident that the two notions agree for $\R^2$, we postpone a general analysis of this equivalence problem to future investigations. 

{\bf Acknowledgements.} We are grateful to G. Marcelli and D. Monaco for interesting discussions on related topics. We gratefully acknowledge financial support from MUR–Italian Ministry of University and Research and Next Generation EU (PRIN project 2022AKRC5P “Interacting Quantum Systems: Topological Phenomena
and Effective Theories”) and from the National Quantum Science and Technology Institute (PNRR MUR project PE0000023-NQSTI).

\section{Preliminaries and main result}

\subsection{The Roe C*-algebra: a quick review}
 Consider the Hilbert space $L^2(\R^d)$, $d\geq 1$, and let $C_c(\R^d)$ be the set of compactly supported continuous functions on $\R^d$. To each $f\in C_c(\R^d)$ one associates a bounded multiplication operator on $L^2(\R^d)$, namely
\begin{equation*}
    (f\ph)(x)=f(x)\ph(x), \quad \forall \ph\in L^2(\R^d).
\end{equation*} 
 We also denote by $\mathcal{B}(L^2(\R^d)$ the set of bounded linear operators on $L^2(\R^d)$ and by $\di$ the distance between sets, which is given by \begin{equation*}
     \di(A,B):=\inf_{x\in A,y\in B}\norm{x-y}, \quad \forall A,B\subseteq \R^d.
\end{equation*}
In order to define the Roe $C^*$-algebra of $\R^d$ we begin with the next definition, following
\cite{Higson Roe, Roe, Coarse geometry, Large scale}.
\begin{dfn} \label{loc comp fin prop}
	Consider $T\in\mathcal{B}(L^2(\R^d))$. We say that
	\begin{enumerate}[label=(\roman*)]
		  \item T is locally compact if $Tf, \ fT$ are compact operators on $L^2(\R^d)$, $\forall f\in C_c(\R^d)$;
		  \item T has finite propagation if there exists $R>0$ such that $fTg=0$, whenever \\ $\di(\hbox{supp}(f),\hbox{supp}(g))>R$, $\forall f,g\in C_c(\R^d)$.
	\end{enumerate}
We define $\C[\R^d]:=\set{T\in \mathcal{B}(L^2(\R^d)): T \text{ is locally compact with finite propagation}}$.
\end{dfn}
One may think of finite propagation operators as the continuum analogous of a matrix, with discrete indices replaced by $(x,y)\in\R^d\times \R^d$, with vanishing entries outside a neighborhood of the diagonal.

\begin{example}[Convolution operator with integral kernel in $C_c(\R^d)$] \,\\
	For $k\in C_c(\R^d)$, define the operator $K\in \BR$ acting on $L^2(\R^d)$ by convolution, \ie
	\begin{equation*}
		  (K\ph)(x)=\int_{\R^d}\di y \, k(x-y)\ph(y), \quad \ph \in L^2(\R^d).
	\end{equation*}
The operator $K$ is locally compact, because for every $f\in C_c(\R^d)$, $Kf,fK$ are still integral operators with compactly supported, continuous integral kernels. To see that $K$ has also finite propagation define $R:=\hbox{sup}\{\norm{x}: k(x)\neq 0\}$. Then, if $f,g\in C_c(\R^d)$, $\di(\hbox{supp}(f),\hbox{supp}(g))>R$, we have that
\begin{equation*}
	(fKg\ph)(x)=\int_{\R^d}\di y \, f(x)k(x-y)g(y)\ph(y)=0, \quad \forall\ph \in L^2(\R^d), x \in \R^d,
\end{equation*}
because if $x\in\hbox{supp}(f)$, $y\in\hbox{supp}(g)$, then $\norm{x-y}\geq \di(\hbox{supp}(f),\hbox{supp}(g))>R$ and thus $k(x-y)=0$. In particular $K\in\C[\R^d]$.
 \endrmk
\end{example}

\begin{lemma} \label{* subalgebra}
	$\C[\R^d]$ is a $*$-subalgebra of $\mathcal{B}(L^2(\R^d))$.
\end{lemma}
\begin{proof}
    The only non-trivial part to prove is that, if $T,S\in\C[\R^d]$ then $TS$ is also in $\C[\R^d]$. Indeed, if $T,S$ are locally compact their composition is too, because the compact operators form an ideal in $\mathcal{B}(L^2(\R^d))$. Now suppose $T,S$ have finite propagation $R_1,R_2$, respectively. Fix $\eps>0$, let $\set{x_n}_{n\in\N}\subseteq \R^d$ be a dense subset and let $\{\psi_n\}_{n\in\N}\subseteq C_c(\R^d)$ be a partition of unity subordinated to $\set{B_{\eps}(x_n)}_{n\in\N}$, where $B_{\eps}(x_n):=\set{x\in \R^d: \norm{x-x_n}<\eps}$. Then, if $f,g\in C_c(\R^d)$, we have that
    \begin{equation*}
		fTSg=fT\(\sum_{n\in N_1\cap N_2}\psi_n\)Sg,
    \end{equation*} 
    with $N_1\!:=\!\{n: \di(\hbox{supp}(\psi_n), \hbox{supp}(f))\leq R_1 \}$, $N_2\!:=\!\{n: \di(\hbox{supp}(\psi_n), \hbox{supp}(g))\leq R_2 \}$.
    Suppose $fTSg\neq 0$. Then there exists $n\in N_1\cap N_2$ and, by Weierstrass Theorem, there exist $z,z'\in\supp(\psi_n),x\in\supp(f), y\in\supp(g)$ such that 
    \begin{equation*}
        \di(\supp(f),\supp(\psi_n))=\norm{x-z}\leq R_1, \quad \di(\supp(g),\supp(\psi_n))=\norm{y-z'}\leq R_2.
    \end{equation*}
    Therefore, since $\supp(\psi_n)\subseteq B_{\eps}(x_m)$ for some $m\in\N$, we obtain
    \begin{equation*}
	\di(\supp(f),\supp(g))\leq \norm{x-y}\leq \norm{x-z}+\norm{z-z'}+\norm{z'-y}\leq R_1+R_2+2\eps.
    \end{equation*}
    By the arbitrarity of $\eps$, it follows that $TS$ has finite propagation $R_1+R_2$.
\end{proof} 
\begin{dfn}
    The Roe $C^*$-algebra $C^*(\R^d)$ is the closure in $\mathcal{B}(L^2(\R^d))$ of $\C[\R^d]$.
\end{dfn}
Observe that $C^*(\R^d)$ is not a unital $C^*$-algebra. Indeed, suppose that the identity operator is the norm-limit of a sequence of locally compact operators $\set{T_n}_{n\in\N}$. Then, for all $f\in C_c(\R^d)$ we would have
\begin{equation*}
    \norm{f-T_nf}\leq \norm{\Id-T_n}\norm{f}\xrightarrow{n\to\infty}0.
\end{equation*}
As a consequence the multiplication operator by $f$ on $L^2(\R^d)$ would be compact for every $f\in C_c(\R^d)$, which is absurd. 

\subsection{Generalized Wannier bases and Roe triviality}

In this section we introduce the concept of generalized Wannier basis. We first consider a discrete set with some uniformity property, which has the role of modelling a non-crystalline material.
\begin{dfn}
    Let $\D$ be a subset of $\R^d$. We say that $\D$ is $r$-uniformly discrete if there exists $r>0$ such that the set $\set{B_r(\ga)}_{\ga\in\D}$ is composed of mutually disjoint balls.
\end{dfn}

\begin{dfn}\label{GWB}
    Let $\Hi$ be a closed subspace of $L^2(\R^d)$. Given $s>0$ we say that $\Hi$ admits a $s$-localized generalized Wannier basis (GWB) if there exists a uniformly discrete set $\D$ and a set $\set{\psi_{\ga}}_{\ga\in\D}\subseteq L^2(\R^d)$ such that
    \begin{enumerate}[label=(\roman*)]
        \item $\set{\psi_{\ga}}_{\ga\in\D}$ is an orthonormal basis for $\Hi$;
        \item  
         the functions $\psi_{\ga}$ are uniformly $s$-localized around $\D$, \ie there exists $M>0$ such that 
        \begin{equation*}
             \int_{\R^d}\di x \, \langle x-\ga \rangle^{2s}|\psi_{\ga}(x)|^2\leq M, \quad \forall \ga\in\D,
        \end{equation*}
    \end{enumerate}
    where we used the standard notation $\langle x\rangle=(1+\norm{x}^2)^{\frac{1}{2}}$.
\end{dfn}

We now consider the topological $K$-theory of the $C^*$-algebra $C^*(\R^d)$. In order to proceed, we first give the following definition of equivalence between projections.

\begin{dfn} 
    Let $\A$ be a $C^*$-algebra. Two projections $P,Q\in \A$ are said to be Murray-von Neumann (MvN) equivalent if there exists $V\in \A$ such that $P=VV^*$ and $Q=V^*V$.
\end{dfn}
By \cite[Cor. 2.20]{Coarse geometry} the elements in $K_0(C^*(\R^d))$ can be represented as formal differences of MvN equivalence classes of projections in $C^*(\R^d)$.
In particular, to each projection $P\in C^*(\R^d)$ we can associate an element $[P]\in K_0(C^*(\R^d))$ which is given by the formal difference of the MvN equivalence classes of $P$ with the zero projection in $C^*(\R^d)$. If such a class defines a trivial element in $K_0(C^*(\R^d))$, we say that the projection $P$ is Roe trivial.

\subsection{Main result}

We are now in position to state our main result, namely that the existence of a $s$-localized GWB for $s>\frac{d}{2}$ implies the Roe triviality of the corresponding projection. 
\begin{thm}\label{original}
    Let $\Hi$ be a closed subspace of $L^2(\R^d)$ admitting a $s$-localized GWB, $\set{\psi_{\ga}}_{\ga\in\D}$, with $s>\frac{d}{2}$, where $\D$ is a uniformly discrete subset of $\R^d$. \newline Then the projection operator $P_{\Hi}=\sum_{\ga\in\D}\ket{\psi_{\ga}}\bra{\psi_{\ga}}$ defines a trivial element in $K_0(C^*(\R^d))$, namely $[P_ {\Hi}]=0$ in $K_0(C^*(\R^d))$.
\end{thm}
Observe that, as mentioned in the Introduction, the Theorem holds for any dimension. In particular, when $d=2$ we almost reach the threshold value $s_*=1$, under which the $s$-localization of a GWB is not expected to imply topological triviality of $P_{\Hi}$. The same threshold appears also in \cite{Lu2}, where the authors proved instead the Chern triviality of the projection $P_{\Hi}$.

\subsection{Application to a physically relevant model}
To show a possible application of our Theorem, we consider here a deformation ``á la Gubanov" of the two-dimensional Kronig-Penney model. The latter is given by the following one-particle Schr\"{o}dinger operator acting in $L^2(\R^2)$
\begin{equation*}
    H_{0}=-\Delta+V_0,
\end{equation*}
where $V_0$ is a two-dimensional $\Z^2$-periodic well potential, \ie
\begin{equation*}
    V_0(x)=v_0\sum_{n\in\Z^2}\chi_{\mathcal{R}(n)}(x),
\end{equation*}
where $v_0\in\R$ is the height of the barrier and  $\mathcal{R}(n):=\mathcal{R}(0)+n=[n_1-\frac{a}{2},n_1+\frac{a}{2}]\\ \times [n_2-\frac{a}{2},n_2+\frac{a}{2}]$, with $0<a<1$ a length parameter. It is known that, for suitable height and length parameters, the spectrum of $H_0$ admits gaps (\cite{Kronig Penney}, for numerical simulations see \cite{Numerics}). We denote by $\sigma_0$ one of the connected components of the spectrum (spectral island), for which it holds
\begin{equation*}
    \dist(\sigma(H_0),\sigma_0)>0.
\end{equation*}
We also denote by $P_0=\chi_{\sigma_0}(H_0)$ the spectral projection onto $\sigma_0$. Since the Hamiltonian $H_0$ is time-reversal symmetric, by \cite[Thm. 1]{Panati2007} and \cite[Thm. 3.5]{Monaco}, there exists an orthornormal (Wannier) basis $\set{w_n}_{n\in\Z^2}=\set{w_0(\cdot-n)}_{n\in\Z^2}$ of $\Ran P_0$ which is exponentially localized, \ie there exists $\alpha_0>0$ such that 
\begin{equation*}
    \int_{\R^2}\di x \, e^{2\alpha_0\norm{x-n}}|w_n(x)|^2< \infty, \quad \forall n\in\Z^2,
\end{equation*}
where we assumed, for simplicity, no degeneracy of the localization centers $n\in \Z^2$. Clearly $\set{w_n}_{n\in\Z^2}$ is also a $s$-localized GWB around $\Z^2$, for each $s>0.$

In order to describe possible deformations of the system, we consider a vector-valued function $g\in C^3(\R^2;\R^2)$ satisfying
\begin{equation*}
    \xi:=\sup_{x\in\R^2}\max_{1\leq i,j,k,\ell\leq 2}\set{|\partial_ig_k(x)|,|\partial_i\partial_j g_k(x)|,|\partial_i\partial_j\partial_{\ell} g_k(x)|}<\infty.
\end{equation*}
We represent the deformed system through the so called quasi-crystalline Gubanov model (\cite{Nenciu, Boutet, Marcelli}), given by the Hamiltonian
\begin{equation*}
    H_g=-\Delta+V_g,
\end{equation*}
where $V_g(x)=V_0(x+g(x))$. By the Hadamard-Caccioppoli Theorem the function $f(x):=x+g(x)$ is globally invertible for $\xi$ small enough, and we denote its inverse by $h(x):=f^{-1}(x).$ For each $\ph\in L^2(\R^2)$, we consider the following operators
\begin{align*}
    & (Y\ph)(x)=J(x)^{\frac{1}{2}}\ph(f(x)),
    & (Y^{-1}\ph)(x)=J(h(x))^{-\frac{1}{2}}\ph(h(x)),
\end{align*}
where $J(x)=|\det Df(x)|=|\det(\Id+Dg(x))|$. By the change of variable formula, it's easy to see that $Y^*=Y^{-1}$, so $Y$ is a unitary operator. By a direct computation one shows that, for $\xi$ small enough
\begin{equation*}
    Y^{-1}H_g Y=H_0+\xi D_{\xi},
\end{equation*}
where $D_{\xi}$ is a $H_0$-bounded operator, \ie there exist $c>0$, $d>0$, independent of $\xi$, such that
\begin{equation*}
    \norm{D_{\xi}\ph}\leq c\norm{H_0\ph}+d\norm{\ph}, \quad \forall \ph\in\mathcal{D}(H_0).
\end{equation*}
Then, by \cite[Thm. VI.5.12]{Kato}, we have that the spectrum of $H_g$ has a spectral island $\sigma_g$ that varies continuously w.r.t the Hausdorff distance, and we denote by $P_g=\chi_{\sigma_g}(H_g)$ its spectral projection onto $\sigma_g.$ As discussed \eg in \cite[Example 2.8]{Marcelli}, thanks to the Combes-Thomas estimate and the Kato-Nagy formula, it is possible to unitarily map the GWB in $\Ran P_0$ to a GWB in $\Ran(Y^{-1}P_gY)$, hereafter denoted by $\set{\psi_n}_{n\in\Z^2}$, without losing its exponential localization around $\Z^2$, \ie there exist $M>0$, $\alpha>0$ such that
\begin{equation*}
    \int_{\R^2}\di x \, e^{2\alpha\norm{x-n}}|\psi_{n}(x)|^2\leq M, \quad \forall n\in\Z^2.
\end{equation*}

The deformed lattice $\D:=h(\Z^2)=\set{\ga\in\R^2:\ga=h(n),\, n\in\Z^2}$ is a $r$-uniformly discrete set,
with $r=(1+2\xi)^{-1}$. By setting
\begin{equation*}
    \ph_{\ga}:=Y\psi_n,\quad \forall \ga=h(n)\in\D,
\end{equation*}
the set $\set{\ph_{\ga}}_{\ga\in\D}$ is a GWB for $\Ran P_g$, which is exponentially localized around $\D$. 
Indeed, notice that for all $x,y\in\R^2$, for $\xi<\frac{1}{2}$, applying the mean value Theorem to each component of $g$, we have 
\begin{equation*}
   \norm{f(x)-f(y)}\geq  (1-2\xi)\norm{x-y}\Longrightarrow \norm{h(x)-h(y)}\leq (1-2\xi)^{-1}\norm{x-y}.
\end{equation*}
Therefore, for all $\ga=h(n)\in\D$, taking $\beta=\alpha(1-2\xi)$ and using the change of variable $x=h(y)$, we obtain
\begin{align*}
    \int_{\R^2}\di x\, e^{2\beta\norm{x-\ga}}|\ph_{\ga}(x)|^2&=\int_{\R^2}\di x\, e^{2\beta\norm{x-h(n)}}J(x)|\psi_n(f(x))|^2 \\
    &=\int_{\R^2}\di y\, e^{2\beta\norm{h(y)-h(n)}}|\psi_n(y)|^2 
    \leq \int_{\R^2}\!\di y\,e^{2\alpha\norm{y-n}}|\psi_n(y)|^2\leq M,
\end{align*}
which proves exponential localization. In particular Theorem \ref{original} applies to $\Hi=\Ran P_g$, since it admits a GWB $\set{\ph_{\ga}}_{\ga\in\D}$ which is also $s$-localized around a $r$-uniformly discrete set $\D$, for all $s>0$. As a consequence we get that $[P_g]=0$ in $K_0(C^*(\R^2))$. 

We stress out that both the Kronig-Penney model and its Gubanov deformation can be generalized to the $\R^d$ case, for any dimension $d\geq 1$. Theorem \ref{original} applies to $d$-generalizations as well.
\section{Proofs}
\subsection{Proof of the main Theorem}
In order to prove our main result, it will be useful to introduce the concept of extremely localized generalized Wannier basis.
\begin{dfn}
    Let $\Hi$ be a closed subspace of $L^2(\R^d)$. We say that $\Hi$ admits an extremely localized GWB if there exists a $r$-uniformly discrete set $\D$ and a set $\set{\phi_{\ga} }_{\ga\in\D}\subseteq L^2(\R^d)$ which is an orthonormal basis for $\Hi$, such that $\hbox{supp}(\phi_{\ga})\subseteq B_r(\ga)$, for all  $\ga \in \D$.
\end{dfn}

A projection operator whose range admits an extremely localized GWB is an element of $C^*(\R^d)$ and its $K$-theory class $[P]$ is trivial, as shown in the next proposition, whose proof is postponed to the next section. 
\begin{prop}\label{extremely localized}
    Let $\Hi$ be a closed subspace of $L^2(\R^d)$. Suppose that $\Hi$ admits an extremely localized GWB $\set{\phi_\ga}_{\ga\in \D}$, where $\D\subseteq \R^d$ is a $r$-uniformly discrete set.  Then 
    \begin{enumerate}[label=(\roman*)]
        \item The associated projection operator $\displaystyle P_{\Hi}:=\sum_{\ga\in\D}\ket{\phi_{\ga}}\bra{\phi_{\ga}}$ belongs to $\C[\R^d]$; \vspace{-3mm}
        \item $[P_ {\Hi}]=0$ in $K_0(C^*(\R^d))$.
    \end{enumerate}
\end{prop}

In light of the previous proposition, we are able to explain the strategy of the proof of Theorem \ref{original}: it consists in showing that a projection whose range admits a sufficiently fast decaying GWB is MvN equivalent to a projection onto the span of some extremely localized GWB. In this way, their $K_0$-classes would define the same element in $K_0(C^*(\R^d))$, which is the trivial one by Proposition \ref{extremely localized} (ii), thus proving the thesis of our main Theorem.
\clearpage

To implement this idea, we need the following lemma, which is a specialization to our setting of \cite[Prop. 3.4]{Large scale}. Again, the proof is postponed to the next section.
\begin{lemma}\label{cauchy maclaurin}
    Let $\D$ be a $r$-uniformly discrete subset of $\R^d$. Then for every $s>\frac{d}{2}$ there exists $C>0$ such that
    \begin{equation*}
		\sum_{\substack{\ga\in\D \\ \norm{x-\ga}\geq R}}\langle x-\ga\rangle^{-2s}\leq C(1+R)^{d-2s}, \quad \forall x\in \R^d, R\geq 0.
    \end{equation*} 
\end{lemma}
Observe that the hypothesis $s>\frac{d}{2}$ is necessary to guarantee the convergence of the series.

Now we are going to prove the key result of the article, following the strategy we have just described.
\begin{prop} \label{MVN}
    Let $\Hi$ be a closed subspace of $L^2(\R^d)$ admitting a $s$-localized GWB, $\set{\psi_{\ga}}_{\ga\in\D}$, with $s>\frac{d}{2}$, where $\D$ is a $r$-uniformly discrete subset of $\R^d$. \newline Then for every closed Hilbert subspace $\widetilde{\Hi}$ admitting an extremely localized GWB $\set{\phi_{\ga}}_{\ga\in\D}$ there exists $V\in C^*(\R^d)$ such that
    \begin{align*}
	P_{\Hi}&:=\sum_{\ga\in\D}\ket{\psi_{\ga}}\bra{\psi_{\ga}}=V^*V, & P_{\widetilde{\Hi}}&:=\sum_{\ga\in\D}\ket{\phi_{\ga}}\bra{\phi_{\ga}}=VV^*.
    \end{align*}
    As a consequence, we have that $P_{\Hi}\in C^*(\R^d)$ and that $P_{\Hi}$ is MvN equivalent to $P_{\tilde{\Hi}}$.
\end{prop}
\begin{proof}
    Let $\widetilde{\Hi}$ be a closed Hilbert subspace of $L^2(\R^d)$ admitting an extremely localized GWB $\set{\phi_{\ga}}_{\ga\in\D}$. Define the operator $V$ by
    \begin{equation*}
		V:=\sum_{\ga\in\D}\ket{\phi_{\ga}}\bra{\psi_{\ga}}.
    \end{equation*}
    It is clear from the definition that $P_{\Hi}=V^*V$, $P_{\widetilde{\Hi}}=VV^*$. Indeed
    \begin{equation*}
        V^*V=\sum_{\ga,\eta\in\D}\ket{\psi_{\ga}}\langle \phi_{\ga}|\phi_{\eta} \rangle\bra{\psi_{\eta}}=\sum_{\ga,\eta\in\D}\ket{\psi_{\ga}}\bra{\psi_{\eta}}\delta_{\ga,\eta}=P_{\Hi},
    \end{equation*}
    where we used the orthonormality property of $\set{\phi_{\ga}}_{\ga\in\D}$. Analogously, using the fact that $\set{\psi_{\ga}}_{\ga\in\D}$ is an orthonormal basis, we have also that $VV^*=P_{\widetilde{\Hi}}$. In order to show that $V\in C^*(\R^d)$ we want to approximate it with operators $V^R\in\C[\R^d]$, with $R$ tending to infinity. 
    More precisely, for every $R\geq 0$, we define
    \begin{equation*}
	V^R:=\sum_{\ga\in\D}\ket{\phi_{\ga}}\bra{\psi^R_{\ga}},
    \end{equation*}
    where $\psi^R_{\ga}=\chi_{B_R(\ga)} \psi_{\ga}$ is the restriction of $\psi_{\ga}$ to the ball $B_R(\ga)$. Clearly, $V^R\in \C[\R^d]$. Indeed, if $f\in C_c(\R^d)$, we have that
    \begin{align*}
	fV^R&=\sum_{\ga\in\D_1}\ket{f\phi_{\ga}}\bra{\psi^R_{\ga}}, & V^Rf&=\sum_{\ga\in\widetilde{\D}_1}\ket{\phi_{\ga}}\bra{\psi^R_{\ga}\bar{f}},
    \end{align*}
    where $\D_1:=\set{\ga:\hbox{supp}(\phi_{\ga})\cap\hbox{supp}(f)\neq \emptyset}$, $\widetilde{\D}_1:=\set{\ga:\hbox{supp}(\psi^R_{\ga})\cap\hbox{supp}(f)\neq \emptyset}$. Since $\D$ is $r$-uniformly discrete and because $\hbox{supp}(\psi^R_{\ga})\subseteq B_R(\ga)$, $\hbox{supp}(\phi_{\ga})\subseteq B_r(\ga)$, for all $\ga\in\D$, we have that both $\D_1$, $\widetilde{\D}_1$ have only a finite number of elements. Thus both $V^Rf$,$fV^R$ are finite rank operators and are thus compact. Considering now $f,g\in C_c(\R^d)$, it holds that
    \begin{equation*}
	fV^Rg=\sum_{\ga\in\D_1\cap\D_2}\ket{f\phi_{\ga}}\bra{\psi^R_{\ga}\bar{g}},
    \end{equation*}
    where $\D_2:=\set{\ga:\hbox{supp}(\psi^R_{\ga})\cap\hbox{supp}(g)\neq \emptyset}$. If $fV^Rg\neq 0$ then there exists $\ga\in\D_1\cap \D_2$. In particular there exist $x\in B_r(\ga)\cap \hbox{supp}(f)$, $y\in B_R(\ga)\cap \hbox{supp}(g)$. Thus we obtain that
    \begin{equation*}
	\di(\hbox{supp}(f),\hbox{supp}(g))\leq \norm{x-y}\leq \norm{x-\ga}+\norm{y-\ga}\leq r+R.
    \end{equation*}
    So $V^R$ has finite propagation $r+R$. Therefore we only need to show that
    \begin{equation}\label{conv}
	\lim_{R\to \infty}\norm{V-V^R}=0,
    \end{equation}
    which implies that $V\in C^*(\R^d)$.
    In order to prove \eqref{conv}, we consider the operator 
    \begin{equation*}
	V-V^R=\sum_{\ga\in\D}\ket{\phi_{\ga}}\bra{\psi_{\ga}-\psi^R_{\ga}},
    \end{equation*}
    which is an integral operator with integral kernel given by
    \begin{equation}\label{kernelz}
	K^R(x,y)=\sum_{\ga\in\D}\phi_{\ga}(x)\overline{\(\psi_{\ga}(y)-\psi^R_{\ga}(y)\)}, \quad x,y\in\R^d.
    \end{equation}
    It is useful to notice that the sum in \eqref{kernelz} is finite, for fixed $x\in\R^d$. Indeed, if $x\in B_r(\eta)$, $\eta\in\D$, then the summands corresponding to $\ga\neq\eta$ give no contribution to the sum, because the supports of $\phi_{\ga}$ are all disjoint. On the other hand, if $x\notin B_r(\ga)$, for any $\ga\in\D$, then $K^R(x,y)=0$, for all $y\in \R^d$.
    \noindent Now, for any $\ph\in L^2(\R^d)$, we have that
    \begin{align}\label{key}
	\norm{(V-V^R)\ph}^2&=\int_{\R^d}\di x \, \left|\int_{\R^d}\di y \, K^R(x,y)\ph(y)\right|^2 \nonumber \\
	&=\int_{\R^d}\di x \, \left|\int_{\R^d}\di y \, \left[\sum_{\ga\in\D}\phi_{\ga}(x)\overline{\(\psi_{\ga}(y)-\psi^R_{\ga}(y)\)}\right]\ph(y)\right|^2 \nonumber\\
	&=\sum_{\eta\in\D}\int_{B_r(\eta)}\di x \, \left|\int_{\R^d}\di y \, \left[\sum_{\ga\in\D}\phi_{\ga}(x)\overline{\(\psi_{\ga}(y)-\psi^R_{\ga}(y)\)}\right]\ph(y)\right|^2 \nonumber\\
	&=\sum_{\ga\in\D}\int_{B_r(\ga)}\di x \, \left|\int_{\R^d}\di y \, \phi_{\ga}(x)\overline{\(\psi_{\ga}(y)-\psi^R_{\ga}(y)\)}\ph(y)\right|^2 \nonumber\\
	&=\sum_{\ga\in\D} \left|\int_{\R^d}\di y \, \overline{\(\psi_{\ga}(y)-\psi^R_{\ga}(y)\)}\ph(y)\right|^2,
    \end{align}
    where in the third equality we used that the function $x\mapsto K^R(x,y)$ is supported in $\cup_{\eta\in\D}B_r(\eta)$, in the second-to-last one we used the fact that, if $x\in B_r(\eta)$, then $\phi_{\ga}(x)=0$, for all $\ga\neq \eta$, and in the last one we used that $\norm{\phi_{\ga}}=1$, $\forall \ga\in\D$.
    \\
    Then, multiplying and dividing by $\langle y-\ga \rangle^s$ inside the integral and using Cauchy-Schwarz inequality and the property of uniform localization of $\set{\psi_{\ga}}_{\ga\in\D}$, we obtain
    \begin{align*}
	\norm{(V-V^R)\ph}^2&=\sum_{\ga\in\D} \left|\int_{\R^d}\di y \, \overline{\psi_{\ga}(y)}\(1-\chi_{B_R(\ga)}(y)\)\frac{\langle y-\ga\rangle^s}{\langle y-\ga \rangle^{s}}\ph(y)\right|^2 \\
	&\leq \sum_{\ga\in\D}\(\int_{\R^d}\di y \, |\psi_{\ga}(y)|^2\langle y-\ga \rangle^{2s}\)\(\int_{\R^d}\di y \,\frac{|\ph(y)|^2}{\langle y-\ga \rangle^{2s}}\(1-\chi_{B_R(\ga)}(y)\)\) \\
	&\leq M\sum_{\ga\in\D}\int_{\R^d}\di y \, |\ph(y)|^2\langle y-\ga \rangle^{-2s}\(1-\chi_{B_R(\ga)}(y)\)\\
	&=M\int_{\R^d}\di y \,  |\ph(y)|^2\sum_{\ga\in\D}\langle y-\ga \rangle^{-2s}\(1-\chi_{B_R(y)}(\ga)\) \\
	&=M\int_{\R^d}\di y \,  |\ph(y)|^2\sum_{\substack{\ga\in\D \\ \norm{y-\ga}\geq R}}\langle y-\ga\rangle^{-2s}\leq C(1+R)^{d-2s}\norm{\ph}^2,
    \end{align*}
    where in the second-to-last line we used Tonelli's Theorem to exchange the series with the integral, and in the last inequality we used the result of Proposition \ref{cauchy maclaurin}. 
    In conclusion we obtained
    \begin{equation*}
	\norm{V-V^R}^2\leq C(1+R)^{d-2s}\xrightarrow{R \to \infty} 0,
    \end{equation*} 
    because $s>\frac{d}{2}$, and this implies that $V\in C^*(\R^d)$.
\end{proof}
The previous argument is similar to the proof of \cite[Prop. 3.4]{Large scale}, where the auxiliary operator $V$
is replaced with its  adjoint, interchanging the role of bra and ket.
Although this exchange has no impact on the MvN equivalence, it makes a difference when estimating $\norm{(V-V^R)\ph}$, since both the test function and the $s$-localized generalized Wannier functions are now calculated in the same integration variable, namely
\begin{equation*}
    \big((V-V^R)\ph\big)(x)=\int_{\R^d}\di y \Bigg[\sum_{\ga\in\D}\phi_{\ga}(x)\overline{\(\psi_{\ga}(y)-\psi^R_{\ga}(y)\)}\ph(y)\Bigg].
\end{equation*}
This little but crucial modification,  allowed us to bring the series appearing in the expression of $K^R$ outside of the integral in \eqref{key}, thanks to the orthonormality property of $\set{\phi_{\ga}}_{\ga\in\D}$, making the subsequent estimates easier. In view of that, we were able to reach an almost-optimal threshold, \ie $s>\frac{d}{2}$.

\begin{proof}[Proof of Theorem \ref{original}]
    By Proposition \ref{MVN} we know that $P_{\Hi}$ is MvN equivalent to a projection operator $P_{\widetilde{\Hi}}$ whose range admits an extremely localized GWB. As a consequence, they have the same class in $K_0(C^*(\R^d))$. Thus, by Proposition \ref{extremely localized}(ii), we conclude that
    \begin{equation*}
		[P_{\Hi}]=[P_{\widetilde{\Hi}}]=0, \quad \hbox{in} \ K_0(C^*(\R^d)).
    \end{equation*}
\end{proof}
\subsection{Proofs of technical results}

\begin{proof}[Proof of Proposition \ref{extremely localized}]
    (i). Let $f\in C_c(\R^d)$, then $fP_{\Hi}, P_{\Hi}f$ are compact operators. Indeed 
    \begin{align*}
		fP_{\Hi}&=\sum_{\ga\in\D_1}\ket{f\phi_{\ga}}\bra{\phi_{\ga}}, & P_{\Hi}f&=\sum_{\ga\in\D_1}\ket{\phi_{\ga}}\bra{\phi_{\ga}\bar{f}},
    \end{align*}
    where $\D_1:=\set{\ga\in\D: \hbox{supp}(\phi_{\ga})\cap \hbox{supp}(f)\neq \emptyset}$. Since $\D$ is $r$-uniformly discrete and $\hbox{supp}(\phi_{\ga})\subseteq B_r(\ga)$, for all $\ga\in\D$, we have that $\D_1$ has only a finite number of elements. This implies that both $fP_{\Hi}, P_{\Hi}f$ are finite-rank operators and are thus compact. To show that $P_{\Hi}$ has also finite propagation consider $f,g\in C_c(\R^d)$. Then we have
    \begin{equation*}
	fP_{\Hi}g=\sum_{\ga\in\D_1\cap \D_2}\ket{f\phi_{\ga}}\bra{\phi_{\ga}\bar{g}},
    \end{equation*}
    where $\D_2:=\set{\ga\in\D: \hbox{supp}(\phi_{\ga})\cap \hbox{supp}(g)\neq \emptyset}$. If $fP_{\Hi}g\neq 0$ then there exists $\ga\in\D_1\cap \D_2$. Specifically there exist $x\in B_r(\ga)\cap \hbox{supp}(f)$, $y\in B_r(\ga)\cap \hbox{supp}(g)$. Thus it holds that
    \begin{equation*}
	\di(\hbox{supp}(f),\hbox{supp}(g))\leq \norm{x-y}\leq \norm{x-\ga}+\norm{\ga-y}\leq 2r.
    \end{equation*}
    This implies that $P_{\Hi}$ has finite propagation $2r$. \\

    \noindent(ii). Consider the decomposition $\R^d=\R^d_{+}\cup \R^d_{-}$, where $\R^d_{\pm}$ are the upper and lower half euclidean spaces. Then, by \cite[Prop. 9.4]{Roe}, we have that $K_0(\R^d_+)=K_0(\R^d_-)=0$, therefore we obtain the thesis using \cite[Prop. 2.5]{Large scale}.
\end{proof}

\begin{proof}[Proof of Lemma \ref{cauchy maclaurin}]
    Let $x\in\R^d$, $R>0$ and $0<\eps<\hbox{min}\set{1,r,R}$. Then, for every $\ga\in \D$ such that $\norm{x-\ga}\geq R$ and for every $y\in B_{\eps}(\ga)$ we have that
    \begin{align}\label{stime}
		\norm{x-y}&\geq \norm{x-\ga}-\norm{\ga-y}>R-\eps, & 1+\norm{x-\ga}\geq 1+\norm{x-y}-\eps>0.
	\end{align}
	Thus we can estimate the series in the following way
	\begin{align}
		\label{serie stima}
		\nonumber \sum_{\substack{\ga\in\D \\ \norm{x-\ga}\geq R}}\(1+\norm{x-\ga}^2\)^{-s}&= \sum_{\substack{\ga\in\D \\ \norm{x-\ga}\geq R}}\frac{1}{|B_{\eps}(\ga)|}\int_{B_{\eps}(\ga)}\di y \, \(1+\norm{x-\ga}^2\)^{-s} \\
		\nonumber &\leq \sum_{\substack{\ga\in\D \\ \norm{x-\ga}\geq R}}\frac{2^s}{|B_{\eps}(\ga)|}\int_{B_{\eps}(\ga)}\di y \, (1+\norm{x-\ga})^{-2s} \\
		&\leq \frac{2^s}{c \, \eps^d}\sum_{\substack{\ga\in\D \\ \norm{x-\ga}\geq R}}\int_{B_{\eps}(\ga)}\di y \, (1+\norm{x-y}-\eps)^{-2s}, 
    \end{align}
    where in the first inequality we used \ $a^2+b^2\geq \frac{1}{2}(a+b)^2$  and in the last one we used the second inequality in $\eqref{stime}$. Now, since the set $\set{B_r(\ga)}_{\ga\in\D}$\! is composed of disjoint balls and because $\eps<r$, we can estimate the series of integrals in $\eqref{serie stima}$ with the integral over all the space. Furthermore, because of the first inequality in \eqref{stime}, we can restrict the integration to the complement of the ball $B_{R-\eps}(x)$.
    Thus we obtain
    \begin{align*}
	\sum_{\substack{\ga\in\D \\ \norm{x-\ga}\geq R}}\(1+\norm{x-\ga}^2\)^{-s}&\leq \frac{2^s}{c \, \eps^d}\int_{\R^d \setminus B_{R-\eps}(x)} \di y \, (1+ \norm{x-y}-\eps)^{-2s} \\
	&= \frac{2^s}{c \, \eps^d}\int_{\R^d \setminus B_{R-\eps}(x)} \di y \, (1+\norm{x-y})^{-2s}\(1-\frac{\eps}{1+\norm{x-y}}\)^{-2s} \\ \\
	&\leq \frac{2^s(1-\eps)^{-2s}}{c \, \eps^d} \int_{\R^d \setminus B_{R-\eps}(x)} \di y \, (1+\norm{x-y})^{-2s} \\ \\
	&= \frac{2^s(1-\eps)^{-2s}}{\hat{c} \, \eps^d} \int_{R-\eps}^{\infty}\di\rho \, \frac{\rho^{d-1}}{(1+\rho)^{2s}} \\ \\
	&\leq \frac{2^s(1-\eps)^{-2s}}{\hat{c} \, \eps^d(d-2s)}(1+R-\eps)^{d-2s}
	=C(1+R)^{d-2s},
    \end{align*}
    where $\displaystyle C=\frac{2^s(1-\eps)^{-2s}}{\hat{c} \, \eps^d(d-2s)}\(1-\frac{\eps}{1+R}\)^{d-2s}$\!\!\!. The proof of the lemma is concluded.
\end{proof}

\bigskip
\flushleft
{\footnotesize
\begin{tabular}{ll}
(V. Rossi)
&  \textsc{GSSI - Gran Sasso Science Institute} \\
&  Viale Francesco Crispi 7, 67100 L'Aquila, Italy \\
&  {E-mail address}: \href{mailto:vincenzo.rossi@gssi.it}{\texttt{vincenzo.rossi@gssi.it}} \\
\\
(G. Panati)
&  \textsc{Dipartimento di Matematica, \virg{La Sapienza} Universit\`{a} di Roma} \\
&  Piazzale Aldo Moro 2, 00185 Rome, Italy \\
&  {E-mail address}: \href{mailto:panati@mat.uniroma1.it}{\texttt{panati@mat.uniroma1.it}} \\
\\
\end{tabular}
}


\end{document}